\documentclass[leqno,12pt]{article}  
\usepackage[margin=1in]{geometry}
\usepackage{mathptmx}

\usepackage{amsthm}
\usepackage{amsmath}
\usepackage{mathbbol}
\usepackage{amssymb}
\usepackage{graphicx}
\usepackage{algorithm}
\usepackage{url}
\usepackage{epsfig}

\usepackage[all]{xy}

\usepackage{xspace}

\usepackage{algorithm}
\usepackage{algpseudocode}

\newcommand{\makespace}{\vspace{0.2cm}}

\newcommand{\R}{\mathbb{R}}

\newcommand{\eps}{\varepsilon}
\newcommand{\compl}{\mathcal{M}}

\newcommand{\im}{\mathrm{Im}}

\newcommand{\dgm}{\mathrm{Dgm}}
\newcommand{\ball}{\mathbb{B}}

\newcommand{\disc}{\Delta} 
\newcommand{\hval}{h}

\newcommand{\Cech}              {\v{C}ech\xspace}
\newcommand{\meb}{\mathrm{meb}}
\newcommand{\mebr}{\mathrm{rad}}
\newcommand{\mebc}{\mathrm{center}}
\newcommand{\diam}{\mathrm{diam}}
\newcommand{\rad}{\mathrm{rad}}

\newcommand{\kapa}{k}

\newcommand{\cech}{\mathcal{C}}
\newcommand{\rips}{\mathcal{R}}
\newcommand{\rep}{\mathrm{rep}}
\renewcommand{\hom}[1]{\hat{#1}}

\newcommand{\Approx}{\mathcal{A}}
\newcommand{\crossmapdown}{\phi}
\newcommand{\crossmapup}{\psi}

\newcommand{\qcell}{\mathrm{qcell}}

\newtheorem{theorem}{Theorem}
\newtheorem{lemma}[theorem]{Lemma}

\newtheorem{definition}[theorem]{Definition}
\newtheorem{observation}[theorem]{Observation}

\numberwithin{equation}{section}
\numberwithin{figure}{section}

\usepackage{color}

\begin{document}

\title{Approximate \v{C}ech Complexes \\in Low and High Dimensions}

\author{Michael Kerber\footnote{Stanford University, Stanford, USA and Max Planck Center for Visual Computing and Communication, Saarbr\"ucken, Germany. \texttt{mkerber@mpi-inf.mpg.de}} \and R.\ Sharathkumar\footnote{Stanford University, Stanford, USA. \texttt{sharathk@stanford.edu}}}

\date{}

\maketitle

\begin{abstract}
\v{C}ech complexes reveal valuable topological information about point sets
at a certain scale in arbitrary dimensions, but the sheer size of these complexes limits 
their practical impact. While recent work introduced approximation techniques
for filtrations of (Vietoris-)Rips complexes, a coarser version of \v{C}ech complexes,
we propose the approximation of \v{C}ech filtrations directly.

For fixed dimensional point set $S$, we present an approximation of the \v{C}ech filtration of $S$ 
by a sequence of complexes of size linear in the number of points. 
We generalize well-separated pair decompositions (WSPD) to well-separated simplicial decomposition (WSSD) 
in which every simplex defined on $S$ is covered by some element of WSSD. 
We give an efficient algorithm to compute a linear-sized WSSD in fixed dimensional spaces. 
Using a WSSD, we then present a linear-sized approximation of the filtration of \v{C}ech complex of $S$. 

We also present a generalization of the known fact that the Rips complex approximates the \v{C}ech complex 
by a factor of $\sqrt{2}$. 
We define a class of complexes that interpolate between \v{C}ech and Rips complexes and that, 
given any parameter $\eps > 0$,
approximate the \v{C}ech complex by a factor $(1+\eps)$. 
Our complex can be represented by roughly $O(n^{\lceil 1/2\eps\rceil})$ simplices 
without any hidden dependence on the ambient dimension of the point set. 
Our results are based on an interesting link between \v{C}ech complex 
and coresets for minimum enclosing ball of high-dimensional point sets. 
As a consequence of our analysis, we show improved bounds on coresets 
that approximate the radius of the minimum enclosing ball.

\end{abstract}

\section{Introduction}

\paragraph{Motivation}
A common theme in topological data analysis is the analysis of point cloud
data representing an unknown manifold. Although the ambient space can
be high-dimensional, the manifold itself is usually of relatively low
dimension. Manifold learning techniques try to infer properties of the
manifold, like its dimension or its homological properties, from the 
point sample. 

An early step in this pipeline is to construct a cell complex from the point
sample which shares similarities with the hidden manifold.
The \emph{\Cech complex at scale $\alpha$} (with $\alpha\geq 0$) captures
the intersection structure of balls of radius $\alpha$ centered at the input
points. More precisely, it is the \emph{nerve} of these balls, and is therefore
homotopically equivalent to their union.
Increasing $\alpha$ from $0$ to $\infty$ yields a \emph{filtration}, 
a sequence of nested \Cech complexes, which can serve as the basis of multi-scale
approaches for topological data analysis.

A notorious problem with \Cech complexes is their representation:
Its $k$-skeleton can consist of up to $O(n^k)$ simplices,
where $n$ is the number of input points.
Moreover, its construction requires the computation of minimum enclosing balls
of point sets; we will make this relation explicit 
in Section~\ref{sec:prelim}. A common workaround is to replace the \Cech
complex by the \emph{(Vietoris-)Rips complex} at the same scale $\alpha$.
Its definition only depends on the diameter of point sets and can therefore
be computed by only looking at the pairwise distances. Although Rips complexes
permit a sparser representation, they do not resolve the issue that the final
complex can consist of a large number of simplices;
Sheehy~\cite{sheehy-linear} and Dey et al.~\cite{dfw-computing} 
have recently addressed this problem by defining 
an approximate Rips filtration whose size is only linear in the input size.
On the other hand, efficient methods for approximating minimum enclosing
balls have been established, even for high-dimensional
problems, whereas the diameter of point sets appears to be a significantly
harder problem in an approximate context. This suggests that \Cech complexes
might be more suitable objects than Rips complexes in an approximate context.

\paragraph{Contribution}
We give two different approaches to approximate filtrations of \Cech complexes,
both connecting the problem to well-known concepts in discrete geometry:
The first approach yields, for a fixed constant dimension, 
a sequence of complexes, each of linear size in the number of input points,
that approximate the \Cech filtration.
By approximate, we mean that the \emph{persistence diagrams} of exact and 
approximate \Cech filtration differ by a arbitrarily small multiplicative factor.
To achieve this result, we generalize the famous 
\emph{well-separated pair decomposition (WSPD)} to a higher-dimensional analogue, that we call 
the \emph{well-separated simplicial decomposition (WSSD)}.
Intuitively, a WSSD decomposes a point set $S$ into $O(n/\eps^d)$ tuples. A $k$-tuple in the WSSD
can be viewed as $k$ clusters of points of $S$ with the property that whenever
a ball contains at least one point of each cluster, a small expansion of the ball contains
all points in all clusters.
Furthermore, these tuples cover every simplex with vertices in $S$, i.e.,
given any $k$-simplex $\sigma$, there is a $k+1$-tuple of clusters such that each cluster contains on vertex of $\sigma$.
We consider the introduction of WSSDs to be of independent interest:
given the numerous applications of WSPD, we hope that its generalization
will find further applications in approximate computational topology.
We finally remark that, similar to related work on the Rips filtration~\cite{sheehy-linear,dfw-computing},
the constant in the size of our filtration
depends exponentially on the dimension of the ambient space, 
which restricts the applicability to low- and medium-dimensional spaces.

As our second contribution, we prove a generalized version of the
well-known Vietoris-Rips lemma~\cite[p.62]{eh-computational} which states
that the \Cech complex at scale $\alpha$ is contained in the Rips complex
at scale $\sqrt{2}\alpha$.
We define a family of complexes, called \emph{completion complexes}
such that for any $\eps$, the \Cech complex at scale $\alpha$
is contained in a completion complex at scale $(1+\eps)\alpha$.
These completions complexes are parametrized by an integer $k$;
the $k$-completion is completely determined by its $k$-skeleton,
consisting of up to $O(n^k)$ complexes.
To achieve $(1+\eps)$-closeness to the \Cech complex, we need to
set $k\approx 1/(2\eps)$ (see Theorem~\ref{thm:interleaving} for the precise statement);
in particular, there is no dependence on the ambient dimension to approximate the \Cech complex
arbitrarily closely.

For proving this result, we use \emph{coresets} for minimum enclosing ball (meb)~\cite{bc-optimal}: 
 the meb of a set of points can be
approximated by selecting only a small subset of the input which is called a \emph{coreset}; 
here approximation means that an $\eps$-expansion of
the meb of the coreset contains all input points.  
The size of the smallest coreset is at most $\lceil 1/\eps\rceil$, independent of the number of points and the ambient dimension, and this bound is tight~\cite{bc-optimal}.
To obtain our result, we relax the definition of coreset for minimum enclosing balls. 
We only require the \emph{radius} of the meb to be approximated, not the meb itself.
We prove that even smaller coresets of size roughly $\lceil 1/(2\eps)\rceil$ always exist 
for approximating the radius of the meb. 
Again, we consider this coreset result to be of independent interest.

\paragraph{Related work}%
Sparse representation of complexes based on point cloud data are 
a popular subject in current research. Standard techniques are the
\emph{alpha complex}~\cite{eks-shape,em-three} which contains all Delaunay
simplices up to a certain circumradius (and their faces), 
\emph{simplex collapses} which remove a pair of simplices from the
complex without changing the homotopy type (see~\cite{als-efficient,mpz-homology,zomorodian-tidy} for modern references), and \emph{witness approaches}
which construct the complex only on a small subset of landmark points
and use the other points 
as witnesses~\cite{sc-topological,bgo-manifold,dfw-graph}.
A more extensive treatment of some of these techniques 
can be found in~\cite[Ch.III]{eh-computational}.
Another very recent approach~\cite{os-zigzag} 
constructs Rips complexes at several scales
and connects them using \emph{zigzag persistence}~\cite{cs-zigzag},
an extension to standard persistence which allows insertions and deletions
in the filtration.
The aforementioned work by Sheehy~\cite{sheehy-linear} combines this theory
with \emph{net-trees}~\cite{hm-fast}, a variant of hierarchical metric spanners,
to get an approximate linear-size zigzag-filtration of the Rips complex
in a first step and finally 
shows that the deletions in the zigzag can be ignored. 
Dey et al.~\cite{dfw-computing} arrive at the same result more directly
by constructing an hierarchical $\eps$-net, defining a filtration from it
where the elements are connected by simplicial maps instead of inclusions,
and finally showing that this filtration is \emph{interleaved} with the 
Rips-filtration in the sense of~\cite{ccggo-proximity}.

\paragraph{Outline}
We will introduce basic topological concepts in Section~\ref{sec:prelim}.
Then we introduce WSSDs, our generalization of WSPDs and give an algorithm
to compute them in Section~\ref{sec:wsd}.
We show how to use WSSDs to approximates the persistence diagram
of the \Cech complex in Section~\ref{sec:linear}.
The existence of small coresets for approximating the radius of the meb is the subject of Section~\ref{sec:coreset}.
$k$-completions and the generalized Vietoris-Rips Lemma are presented in Section~\ref{sec:clique}. 
We conclude in Section~\ref{sec:conclusion}.

\section{Preliminaries}
\label{sec:prelim}

\paragraph{Simplicial complexes}
Let $S$ denote a finite set of universal elements, called \emph{vertices}
\footnote{Some of the defined concepts do not require that $S$ is finite;
however, since we will only deal with finite complexes in later sections,
we decided to discuss this simpler setup.}
A \emph{(simplicial) complex} $C$ is a collection of subsets of $S$,
called \emph{simplices}, with the property that whenever
a simplex $\sigma$ is in $C$, all its (non-empty) subsets are in $C$ as well.
These non-empty subsets are called the \emph{faces} of $\sigma$; a \emph{proper face}
is a face that is not equal to $\sigma$.
Setting $k:=|\sigma|-1$, where $|\cdot|$
stands for the number of elements considered as a subset,
we call $\sigma$ a \emph{$k$-simplex}.
For a $k$-simplex $\sigma=\{v_0,\ldots,v_k\}$, we call $v_0,\ldots,v_k$
its \emph{boundary vertices} of $\sigma$; 
we will also frequently write $\sigma$ as a tuple of its boundary vertices, 
that is, $\sigma=(v_0,\ldots,v_k)$ with the
convention that any permutation of the boundary vertices yields the same
simplex.
A \emph{subcomplex} of $C$ is a simplicial complex that is contained in $S$.
One example of a subcomplex is the \emph{$k$-skeleton} of a complex $C$,
which is the set of all $\ell$-simplices in $C$ with $\ell\leq k$.
Let $K$ and $K'$
be two simplicial complexes with vertex sets $V$ and $V'$ and consider a map $f:V\rightarrow V'$. If for
any simplex $(v_0,\ldots,v_k)$ of $K$, $(f(v_0),\ldots f(v_k))$ yields a simplex in $K'$, then $f$
extends to a map from $K$ to $K'$ which we will also denote by $f$; in this case, $f$
is called a \emph{simplicial map}.

Let $S$ be a set of arbitrary geometric objects, embedded in
an ambient space $\R^d$. 
We call $|S|:=\cup_{s\in S} s\subset\R^d$ the \emph{union of $S$}.
We define a simplicial complex $C$ as follows:
A $k$-simplex $\sigma$ is in $C$ 
if the corresponding $k+1$ objects have a common intersection in $\R^d$.
It is easy to check that $C$ is indeed closed under face relations 
and thus a simplicial complex with vertex set $S$, called the \emph{nerve} of $S$.
The famous \emph{Nerve Theorem}~\cite[p.59]{eh-computational} states
that if all objects in $S$ are convex,
the union of $S$ and its nerve are \emph{homotopically equivalent}.
This intuitively means that one can transform one into the other by
bending, shrinking and expanding, but without gluing and cutting.
A consequence of this theorem is that the \emph{homology groups}
of the union and the nerve are equal. We will give an intuitive
meaning of homology groups later in this section; 
see~\cite{eh-computational,munkres} for thorough introductions to homology.

For a finite point set $P$ and $\alpha>0$, 
the \emph{\Cech complex} $\cech_\alpha(P)$
is the nerve of the set of (closed) balls of radius $\alpha$ centered at the points in $P$.
Note that a $k$-simplex of the \Cech complex 
can be identified with $(k+1)$ points $p_0,\ldots,p_k$ in $P$, the centers
of the intersecting balls.
Let $\meb(p_0,\ldots,p_k)$ denote the \emph{minimum enclosing ball of $P$}, that is, 
the ball with minimal radius that contains each $p_i$. 
\begin{observation}
A $k$-simplex $\{p_0,\ldots,p_k\}$ is in $\cech_\alpha(P)$ iff
the radius of $\meb(p_0,\ldots,p_k)$ is at most~$\alpha$.
\end{observation}

A widely used approximation of \Cech complexes is the \emph{(Vietoris)-Rips complex} $\rips_\alpha(P)$.
It is defined as the maximal simplicial complex whose $1$-skeleton equals the $1$-skeleton
of the \Cech complex.
Described as an iterative construction, starting with the edges of the \Cech complex, 
a triangle is added to the Rips complex when its three
boundary edges are present, a tetrahedron when its four boundary triangles are present, and so forth.
The Rips complex is an example of a \emph{clique complex} (also known as \emph{flag complex} or \emph{Whitney complex}).
That means, it is completely determined by its $1$-skeleton which in turn
only depends on the pairwise distance between the input points. 
For $k+1$ points $p_0,\ldots,p_k$ in $P$, let the \emph{diameter} $\diam(p_0,\ldots,p_k)$
denote the maximal pairwise distance between any two points $p_i$ and $p_j$ with $0\leq i\leq j\leq k$.
\begin{observation}
A $k$-simplex $\{p_0,\ldots,p_k\}$ is $\rips_\alpha(P)$ iff
$\diam(p_0,\ldots,p_k)$ is at most~$\alpha$.
\end{observation}
For notational convenience, we will often omit the $P$ from the notation
and write $\cech_\alpha$ and $\rips_\alpha$ when $P$
is clear from context.

\paragraph{Persistence modules}
For $A\subset\R$, a \emph{persistent module} is a family $(F_\alpha)_{\alpha\in A}$ of vector spaces with homomorphisms
$f_\alpha^{\alpha'}:F_{\alpha}\rightarrow F_{\alpha'}$ for any $\alpha\leq \alpha'$ such that
$f_{\alpha'}^{\alpha''}\circ f_{\alpha}^{\alpha'}=f_{\alpha}^{\alpha''}$
and $f_{\alpha}^{\alpha}$ is the identity function.\footnote{This is not the most general definition of a persistent module;
see~\cite{ccggo-proximity}.}
The most common class are modules induced by a \emph{filtration}, that is, a family of complexes
$(C_\alpha)_{\alpha\in A}$ such that $C_\alpha\subseteq C_{\alpha'}$ for $\alpha\leq\alpha'$. 
For some fixed dimension $p$, set $H_\alpha:=H_p(C_\alpha)$, the $p$-th homology group of $C_\alpha$.
The inclusion map from $C_\alpha$ to $C_{\alpha'}$ induces an homomorphism  
$\hom{f}_{\alpha}^{\alpha'}:H_\alpha\rightarrow H_{\alpha'}$
and turns $(H_\alpha)_{\alpha\in\R}$ into a persistence module. Example of such filtrations and their induced modules
are the \emph{\Cech filtration} $(\cech_\alpha)_{\alpha\geq 0}$ 
and the \emph{Rips filtration} $(\rips_\alpha)_{\alpha\geq 0}$.
However, we will also consider persistence modules which are not induced by filtrations.
Generalizing the case of filtrations, given a sequence of simplicial complexes $(\Approx_\alpha)_{\alpha\in A}$ connected
by simplicial maps $g_\alpha^{\alpha'}:\Approx_\alpha\rightarrow\Approx_{\alpha'}$
which satisfy $g_{\alpha'}^{\alpha''}\circ g_{\alpha}^{\alpha'}=g_{\alpha}^{\alpha''}$ and 
$g_{\alpha}^{\alpha}=\mathrm{id}$, the induced homology groups
$H_\alpha:=H_p(\Approx_\alpha)$ and induced homomorphisms $\hom{g}_{\alpha}^{\alpha'}:H_\alpha\rightarrow H_{\alpha'}$
also yield a persistence module.
A persistence module $(F_\alpha)_{\alpha\in A}$ is \emph{tame}
if the rank of $F_\alpha$ is finite for all $\alpha\in A$. 
As our modules in this work will consist only of homology groups over finite simplicial complexes, 
all modules constructed in this paper will be tame, and we will ignore this technicality from now on
when referring to previous results.
We will frequently denote filtrations and modules by $F_\ast$ instead of $(F_\alpha)_{\alpha\in A}$ for brevity if
there is no confusion about~$A$.

For a persistence module $F_\ast$ with homomorphisms $f_\alpha^{\alpha'}$, 
we say that a generator (basis element) $\gamma\in F_\alpha$ is \emph{born} at $\alpha$ 
if $\gamma\notin \im f_{\alpha-\eps}^{\alpha}$
for any $\eps>0$, where $\im$ is the image of a map.
If $\gamma$ is born at $\alpha$, we say that it \emph{dies} at $\alpha'$
if $\alpha'$ is the smallest value such that $f_{\alpha}^{\alpha'}(\gamma)\in \im f_{\alpha-\eps}^{\alpha'}$ for some $\eps>0$.
In other words, every generator can be represented by a point in the plane, determining its birth- and death-coordinate.
$F_\ast$ is completely characterized by this multiset of points,
which is called the \emph{persistence diagram} of the module and denote it as $\dgm F_\ast$. Note that all points of the diagram lie on or above the
diagonal in the birth-death-plane. 

For the benefit of readers inexperienced with the concept of persistence, 
we explain the wealth of geometric-topological information contained in the persistence diagram,
exemplified on a \Cech filtration of a point set $S$ in $\R^3$. As discussed, we can visualize
the filtration as a sequence of growing balls centered at the points in $S$, and the union 
of these balls forms a sequence of growing shapes.
During this process, the shape might create \emph{voids}, that is, pockets of air completely enclosed by the shape.
The rank of the second homology group $H_2(\cech_\alpha)$ yields the number of voids present 
at a fixed scale $\alpha$
(this rank is also called the $2$nd \emph{Betti number}). The persistence diagram for $H_2(\cech_\ast)$
provides multi-scale information about the voids in the process: every point $(b,d)$ of the diagram represents
a void that is formed for $\alpha=b$ and filled up for $\alpha=d$. 
The same information as for voids can be obtained for \emph{connected components} and for \emph{tunnels},
choosing the $0$th and $1$st homology groups, respectively.

\paragraph{Approximating persistence diagrams}
An important property of persistence diagrams is their stability under ``small'' perturbations
of the underlying filtrations and modules; see Cohen-Steiner et al.~\cite{ceh-stability} for the precise first statement
of this type. We will use the more recent results by Chazal et al.~\cite{ccggo-proximity} for this work,
following Sheehy's notations and definitions~\cite{sheehy-linear}. For two modules $F_\ast$, $G_\ast$, we say that $\dgm\, F_\ast$ is
a \emph{$c$-approximation} of $\dgm\, G_\ast$ with $c\geq 1$ if there is a bijection $\pi:\dgm\, F_\ast\rightarrow\dgm\, G_\ast$
such that for any point $(x,y)$ of $\dgm\, F_\ast$, $\pi(x,y)$ lies in the axis-aligned box defined by $\frac{1}{c}(x,y)$
and $c(x,y)$. An equivalent statement is that the two diagrams have a bounded bottleneck distance on the log-scale.

We will use the following result which is a reformulation of~\cite[Def.4.2+Thm.44]{ccggo-proximity}:

\begin{theorem}\label{thm:c-approx-thm}
Let $(F_\alpha)_{\alpha\geq 0}$ and $(G_\alpha)_{\alpha\geq 0}$ be two persistence module with two families of homomorphisms
$\{\crossmapdown:F_\alpha\rightarrow G_{c\alpha}\}_{\alpha\geq 0}$ and $\{\crossmapup:G_\alpha\rightarrow F_{c\alpha}\}_{\alpha\geq 0}$
such that all the following diagrams commute:
\begin{eqnarray}
\xymatrix{
F_{\frac{\alpha}{c}} \ar[rrr]\ar[rd] & & & F_{c\alpha'}    &  & F_{c\alpha} \ar[r] & F_{c\alpha'} \\
& G_\alpha \ar[r] & G_{\alpha'} \ar[ru] &                 & G_\alpha \ar[r] \ar[ru] & G_{\alpha'} \ar[ru]\\ 
& F_\alpha \ar[r] & F_{\alpha'} \ar[rd] &                 & F_\alpha \ar[r] \ar[rd] & F_{\alpha'} \ar[rd] \\ 
G_{\frac{\alpha}{c}} \ar[rrr]\ar[ru] & & & G_{c\alpha'}    &  & G_{c\alpha} \ar[r] & G_{c\alpha'}\\
}
\label{eqn:commuting}
\end{eqnarray}
Then, the persistence diagrams of $F_\alpha$ and $G_\alpha$ are $c$-approximations of each other.
\end{theorem}
In the case of modules induced by filtrations, there is a simple corollary, called the ``Persistence Approximation Lemma'' in \cite{sheehy-linear}:
\begin{lemma}\label{lem:pat}
If two filtrations $(A_\alpha)_{\alpha\geq 0}$ and $(B_\alpha)_{\alpha\geq 0}$ satisfy $A_\frac{\alpha}{c}\subset B_\alpha\subset A_{c\alpha}$
for all $\alpha\geq 0$, then the persistence diagrams are $c$-approximations of each other.
\end{lemma}

\section{Well-separated simplicial decompositions}
\label{sec:wsd}

In this section, we introduce the notion of Well-separated simplicial decomposition (WSSD) of point sets. 
WSSD can be seen as a generalization of well-separated pair decomposition of a point set. 
We first revisit the definition of WSPD and then generalize it to WSSD. 

\paragraph{Notations.}

Let $S \subset \R^d$ be a fixed point set 
with minimal distance $1/\sqrt{d}$ between two points
and such that all points are contained in a axis-parallel hypercube $q$ with side length $2^L$. 
We consider a \emph{quadtree} $Q$ of $q$ where each node represents a hypercube;
the root represents $q$, and when an internal node represents a hypercube $q'$, 
its children represent the hypercubes obtained by splitting $q'$ into $2^d$ congruent hypercubes. 
From now on, we will usually identify the quadtree node and the hypercube that it represents.
We call a node of $Q$ \emph{empty} if it does not contain any point of $S$.
For any internal node $q'$, the \emph{height} of $q'$ in $Q$ is $i$ if the side length of $q'$ is $2^i$; 
the construction ends at height $0$;
by construction, each leaf contains at most one point of $S$.%
\footnote{This ``construction'' is only conceptual; in an actual implementation, only non-empty would be
stored. Moreover, the quadtree should be represented in \emph{compressed} form to avoid dependence
on the \emph{spread} of the point set; see~\cite[\S 2]{hp-geometric} for details.}
The nodes of $Q$ at height $i$ induce a grid $G_i$ where the side length of every cell of $G_i$ is $2^i$.
For $e>0$ and a ball $\ball$ with center $c$ and radius $r$, we let $e\ball$ denote the ball with center $c$
and radius $e\cdot r$. We state the following property, which follows directly by triangle inequality,
but is used several times in our arguments:
\begin{observation}
\label{obs:offset}
Let $\ball$ be a ball with radius $r$ that intersects a convex object $M$ 
whose diameter is at most $\lambda r$ with some $\lambda>0$.
Then, $M\subseteq \lambda\ball$.
\end{observation}

Finally, whenever we make statements that depend on a parameter $\eps$, it is implicitly
assumed that $\eps\in(0,1)$ from now on.

\paragraph{Well-Separated Pair Decomposition.}
Let $Q$ be a quadtree for $S$. A pair of quadtree cells $(q,q')$ is called \emph{$\eps$-well separated} 
if $\max(\diam(q),\diam(q'))\leq \eps d(q,q')$; 
here $\diam(q)$ is the diameter of a quadtree cell (which equals $2^h\sqrt{d}$ if $h$ is the height of $q$) 
and $d(q,q')$ is the closest distance between cells $q$ and $q'$. 
We state a simple consequence
which appears somewhat indirect, but allows a generalization to multivariate tuples:

\begin{lemma}
\label{lem:wspd-prop}
If $(q,q')$ is $\eps$-well separated, any ball $\ball$ that contains at least one point of 
$q$ and one point of $q'$, the ball $(1+2\eps)\ball$ contains all of $q$ and all of $q'$.
\end{lemma}

\begin{proof}
Let $\ball$ be a ball with radius $r$ intersecting both $q$ and $q'$,
which means that $r\geq d(q,q')/2$.
Because $(q,q')$ is well-separated,
$$\diam(q)\leq \eps d(q,q')\leq 2\eps r,$$
implying that $(1+2\eps)\ball$ contains all of $q$ by Observation~\ref{obs:offset}.
The same argument applies for $q'$.
\end{proof}

For a pair $(p,p')\in S\times S$ we say
that a pair of quadtree cells $(q,q')$ \emph{covers} $(p,p')$ if $p\in q$ and $p'\in q'$, or $p \in q'$ and $p' \in q$.
An \emph{$\eps$-well separated pair decomposition} ($\eps$-WSPD) of $S$ is a set of pairs $\Gamma = ((q_1,q_1'), (q_2,q_2'),\ldots, (q_m,q_m'))$
such that all pairs are $\eps$-well separated and every edge in $S\times S$ is covered by some pair in $\Gamma$.
We rely on the following property of WSPDs, proved first in~\cite{ck-decomposition};
see also \cite[\S 3]{hp-geometric} for a modern treatment:
\begin{theorem}
\label{thm:wspd}
A $\eps$-WSPD of size $O(n/\eps^d)$ can be computed in $O(n \log n+ n/\eps^d )$ time.
\end{theorem}

\paragraph{Well-Separated Simplicial decomposition.}
We generalize the construction of WSPD to higher dimensions: Let $S$ and $Q$ be as above. 
We call a $(k+1)$-tuple $(q_0,\ldots,q_k)$ of quadtree cells an \emph{$\eps$-well separated tuple} ($\eps$-WST),
if for any ball $\ball$ that contains at least one point of each $q_\ell$, we have that
\begin{equation}
q_0\cup q_1 \cup \ldots q_k \subseteq (1+\eps)\ball \label{eq:weight}.
\end{equation}
Moreover, we say that $(q_0,\ldots,q_k)$ \emph{covers} a $k$-simplex $\sigma  = (p_0,\ldots,p_k)$, $p_0,\ldots,p_k \in S$ if there is a permutation $\pi$ of $(0,\ldots, k)$ such that $p_{\pi(\ell)}\in q_\ell$ for all $0\leq\ell\leq k$.

\begin{definition}
A set of $(k+1)$-tuples $\Gamma=\{\gamma_1,\ldots,\gamma_m\}$ is a \emph{$(\eps,k)$-well separated simplicial decomposition} ($(\eps,k)$-WSSD), if
each $\gamma_\ell$ is a $\eps$-well separated tuple and each $k$-simplex of $S$ is covered by some $\gamma_\ell$.
An \emph{$\eps$-WSSD} is the union of $(\eps,\kapa)$-WSSDs over all $1\leq\kapa\leq d$.
\end{definition}

It is easy to see with that an $\frac{\eps}{2}$-WSPD is an $(\eps,1)$-WSSD.

\paragraph{Our algorithm.} 
We present a recursive algorithm for computing an $(\eps,k)$-WSSD. 
If $k=1$, we use the algorithm from \cite[Fig.~3.3]{hp-geometric} to compute an $\frac{\eps}{2}$-WSPD,
which is an $(\eps,1)$-WSSD.
If $k>1$, we recursively compute an $(\eps, \kapa-1)$-WSSD $\Gamma_{\kapa-1}$
and construct an $(\eps,\kapa)$-WSSD $\Gamma_{\kapa}$ as follows:
We initialize $\Gamma_{\kapa}$ as the empty set and iterate over the elements in $\Gamma_{\kapa-1}$. 
For an $\eps$-WST $\gamma=(q_0,q_1,\ldots q_{\kapa-1}) \in \Gamma_{\kapa-1}$, let $\ball_\gamma = \meb(q_0 \cup q_1 \cup \ldots q_{\kapa-1})$,
and let $r$ denote its radius. Consider the grid $G_h$ formed by all quadtree cells of height $h$ 
such that $2^{h} \le \frac{\eps r}{2\sqrt{d}} \le 2^{h+1}$. 
We compute the set of non-empty quadtree cells 
in $G_h$ that intersect the ball $2\cdot\ball_\gamma$.
For each such cell $q'$, we add the $(\kapa+1)$-tuple $(q_0,\ldots,q_{\kapa-1},q')$ to $\Gamma_{\kapa}$.
See Figure~\ref{fig:wssd_illu} for an illustration.

\begin{figure}[thb]
\centering
\includegraphics[width=5cm]{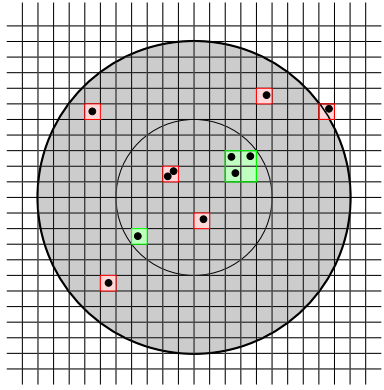}
\caption{
Example for the construction of $\Gamma_{2}$ from $\Gamma_1$: Let the pair of green boxes be a WST $\gamma$ of $\Gamma_1$
(that is, a well-separated pair). Now, the algorithm creates a triple consisting of the two green boxes and any grid cell
at height $h$ that intersects $2\ball_\gamma$ (shaded area). In this example, there would be $10$ triples - $6$ with
the red boxes, and $4$ additional ones coming from the non-empty boxes in the green areas.}
\label{fig:wssd_illu}
\end{figure}

\paragraph{Correctness.}
In order to prove the correctness of our construction procedure, we need to show that the generated tuples indeed form a $(\eps,\kapa)$-WSSD.

\begin{lemma}
\label{lem:wssd_separation}
Every tuple added by our procedure is an $\eps$-WST.
\end{lemma}
\begin{proof}
We do induction on $\kapa$, noting that for $\kapa=1$, the statement is true because an $\frac{\eps}{2}$-WSPD
is an $(\eps,1)$-WSSD.
For $\kapa\geq 2$, assume that our algorithm creates a $\kapa$-tuple $(q_0,\ldots, q_{\kapa - 1}, q')$
by adding the cell $q'$ while considering the $\eps$-WST $(q_0,\ldots, q_{\kapa - 1})$.
Let $\ball$ be a ball that contains at least one point from each of the cells $(q_0,\ldots, q_{\kapa - 1}, q')$. 
We have to argue that $(1+\eps)\ball$ contains the cells $q_0,\ldots, q_{\kapa -1},q'$;
by induction hypothesis, it is clear that $q_0\cup \ldots \cup q_{\kapa -1} \subseteq (1+\eps)\ball$
and moreover, 
\begin{equation*}
r = \mebr(q_0,\ldots,q_{\kapa-1}) \le (1+\eps)\mebr(\ball).
\end{equation*}
Finally, by construction,
$$\diam(q') \le \frac{\sqrt{d}\eps r}{2\sqrt{d}}\le \frac{\eps (1+\eps)\mebr(\ball)}{2}\le \eps\cdot\mebr(\ball),$$
so $q'\subseteq (1+\eps)\ball$ by Observation~\ref{obs:offset}.
\end{proof}

For showing that all $\kapa$-simplices are covered, we use the following result
which is taken from~\cite{bc-optimal}~--
we note that the required bound also follows as a simple corollary of
the main result of~\cite{bc-optimal}, but we decided to give a more low-level argument
for clarity.

\begin{lemma}
\label{lem:meb_center_lemma}
Let $P$ be a point set with $|P|\geq 3$. Then, there exists a point $p\in P$ such
that 
$$p\in \frac{1+1/d}{\sqrt{1-1/d^2}} \meb(P\setminus\{p\}).$$
In particular, $p\in 2\meb(P-\setminus\{p\})$ for $d\geq 2$.
\end{lemma}
\begin{proof}
Note that the statement is trivial if there exists a point $p\in P$ whose removal does
not change the minimum enclosing ball. Therefore, assume wlog that $|P|\leq d+1$,
and all points of $P$ are at the boundary of $\meb(P)$.
Let $c$ be the center and $r$ be the radius of $\meb(P)$.
The points in $P$ span a polytope $T$; take the smallest ball $\ball$ centered at $c$
that is contained in $T$. By~\cite[Lem.~3.2]{bc-optimal}, its radius is at most $r/d$.
Moreover, $\ball$ touches at least one facet of $T$. Let $p$ be the point opposite of this facet,
set $P':=P\setminus\{p\}$ and let $c'$ and $r'$ denote the center and radius of the meb of $P'$. 
Following the argumentation of~\cite[Lem.~3.3]{bc-optimal}, it holds that
$$r'\geq r\sqrt{1-(1/d^2)}$$
and moreover, $c'$ is the point where $\ball$ touches the facet, so that $\|c-c'\|\leq r/d$.
Now, by triangle inequality
\begin{eqnarray*}
\|p-c'\| & \leq & \|p-c\| + \|c-c'\|\\
& \leq & r + r/d\\
& \leq & (1+1/d)\frac{r'}{\sqrt{1-1/d^2}}
\end{eqnarray*}
which implies the first claim. The second part follows easily by noting that
$$\frac{1+1/d}{\sqrt{1-1/d^2}}\leq 2$$
for all $d\geq 5/3$.
\end{proof}

\begin{lemma}
\label{lem:wssd_covering}
The set of $(\kapa+1)$-tuples $\Gamma_\kapa$ generated by our procedure covers all $\kapa$-simplices over $S$.
\end{lemma}
\begin{proof}
We do induction on $\kapa$. For the base case $\kapa=1$, 
by definition, all pairs of points in $S\times S$ are covered by some pair $(q,q')$ in an $\frac{\eps}{2}$-WSPD. 
Assume that the computed $(\eps,\kapa-1)$-WSSD covers all $(\kapa-1)$-simplices
and consider any $\kapa$-simplex $\sigma = (p_0,\ldots,p_k)$. 
By Lemma~\ref{lem:meb_center_lemma}, there exists a point among the $p_i$, say $p_0$, such that
$p_0\in 2\meb(\sigma')$, where $\sigma'=(p_1,\ldots,p_k)$.
By induction hypothesis, there exists a $\eps$-WST $t=(q_1,\ldots,q_\kapa)$ that covers $\sigma'$.
Clearly, $p_0\in 2\meb(t)$ as well.  Let $q$ be the cell of $G_h$ that contains $p_0$. 
By construction, our algorithm adds $(q_1,\ldots,q_{\kapa},q)$ to $\Gamma_\kapa$, 
and this tuple covers $\sigma$.
\end{proof}

With Lemma~\ref{lem:wssd_separation} and Lemma~\ref{lem:wssd_covering}, it follows that
the constructed set $\Gamma_{\kapa}$ is an $(\eps,\kapa)$-WSSD.

\paragraph{Analysis.}
We bound the size of the $(\eps, \kapa)$-WSSD generated by our algorithm and the total time taken to compute it.
\begin{lemma}
\label{lem:size}
Let $\Gamma_\kapa$ be the $(\eps,\kapa)$-WSSD generated by our algorithm. 
Then, $|\Gamma_\kapa| = n (d/\eps)^{O(d\kapa)}$.
\end{lemma}
\begin{proof}
By Theorem~\ref{thm:wspd}, the size of the $(\eps,1)$-WSSD (or $\frac{\eps}{2}$-WSPD) is $O(n(d/\eps)^{O(d)})$. 
Let us assume that the size of $\Gamma_{\kapa-1}$ is $O(n(d/\eps)^{O(d(\kapa-1))})$. 
It suffices to show that for every $\gamma \in \Gamma_{\kapa -1}$, we add at most $O((d/\eps)^d)$ $\eps$-WSTs 
to $\Gamma_\kapa$.

As in the algorithm, set $\ball_\gamma := \meb(\gamma)$ and $r:=\mebr(\gamma)$. 
By construction, the side length of a cell in $G_h$ is at least $\frac{\eps r}{4\sqrt{d}}$. By a simple packing argument, the total number of cells of $G_h$ that intersect $2\ball_\gamma$ is $O((d/\eps)^d)$. 
We add (at most) one $\eps$-WST to $\Gamma_\kapa$ 
for each of these $O((d/\eps)^d)$ cells.
\end{proof}

By Theorem~\ref{thm:wspd}, an $\frac{\eps}{2}$-WSPD can be constructed in $O(n \log n + n (d/\eps)^d)$ time.
To construct $\Gamma_{\kapa}$ from $\Gamma_{\kapa-1}$,
for every $\gamma \in \Gamma_{\kapa-1}$, our algorithm 
has to compute the meb $\ball_\gamma$ of the involved cells 
and find all cells in $G_h$ that intersect $2\ball_\gamma$. This can be done, for instance,
by finding the cell $q$ that contains the center of $\ball_\gamma$ 
and traverse the cells in increasing distance from $q$. All these operations can be done
in time proportional to the number of cells visited, and a constant that only depends on $d$.
Since the total number of visited cells is at most $O((d/\eps)^d)$,
the running time of computing $\Gamma_\kapa$ from $\Gamma_{\kapa-1}$ is $O(|\Gamma_{\kapa-1}|(d/\eps)^d)=O(n (d/\eps)^{O(d\kapa)})$.
It follows that the total running time for computing $\Gamma_1,\ldots,\Gamma_\kapa$ is bounded by
$O(n\log n + n(d/\eps)^{O(d\kapa)})$.
 
We end the section with a property of our computed WSTs which will be useful in Section~\ref{sec:linear}.

\begin{lemma}
\label{lem:heightbnd}
For any $\eps$-WST $t=(q_0,\ldots,q_\kapa)$ generated by our algorithm, 
let $\rho =\mebr(t)$. Then, the height $\lambda$ of each $q_i$ satisfies:
$$2^{\lambda}\leq \frac{\eps \rho}{\sqrt{d}}.$$
\end{lemma} 
\begin{proof}
We do induction on $\kapa$. For $\kapa=1$,
every pair $(q,q') \in \Gamma_1$ is an $\frac{\eps}{2}$-well separated pair. 
With $\ell:=d(q,q')$ the minimum distance between $q$ and $q'$, it is clear that $\rho \ge \ell/2$. 
From the well-separated property, we know that $\max(\diam(q),\diam(q')) \le \frac{\eps\ell}{2}$ 
and therefore, the maximum height $\lambda$ of $q$ and $q'$ is such that 
$2^\lambda \le \frac{\eps\ell}{2\sqrt{d}} \le \frac{\eps\rho}{\sqrt{d}}$ as required.

For $\kapa>1$,
assume that for every $(\eps,\kapa-1)$-WST, the lemma holds. 
Let $\gamma' = q_0,\ldots,q_{\kapa-1} \in \Gamma_{\kapa-1}$ be any $(\eps,\kapa-1)$-WST and $\rho' =\mebr(\gamma')$. 
Assume that our algorithm generates  $\gamma = (q_0,\ldots,q_{\kapa-1}, q')$;
then $q'$ is a cell of level $h$ with $2^{h} \le \frac{\eps\rho'}{2\sqrt{d}}$.
Because $\rho = \mebr(\gamma) \ge \rho'$, this implies that the statement is true for $q'$,
and also holds for $q_0,\ldots,q_{\kapa-1}$ by induction hypothesis.
\end{proof}

\section{\Cech approximations of linear size}
\label{sec:linear}

In this section, we will define a persistence module
which is a $(1+\eps)$-approximation of the \Cech module
in the sense of Section~\ref{sec:prelim}. 
We start with a summary of our construction:
we first define a sequence of (non-nested)
simplicial complexes $(\Approx_\alpha)_{\alpha\geq 0}$,
which we define using a WSSD from Section~\ref{sec:wsd}. Then, we construct
simplicial maps $g_{\alpha}^{\alpha'}:\Approx_\alpha\rightarrow\Approx_{\alpha'}$
such that $g_{\alpha'}^{\alpha''}\circ g_{\alpha}^{\alpha'}=g_{\alpha}^{\alpha''}$ and 
$g_{\alpha}^{\alpha}=\mathrm{id}$.
As discussed in Section~\ref{sec:prelim}, applying the homology functor to that sequence
yields a persistent module.
To show that the constructed module approximates the \Cech module, we
define simplicial \emph{cross-maps} $\crossmapdown:\cech_{\frac{\alpha}{1+\eps}}\rightarrow\Approx_{\alpha}$
and $\crossmapup:\Approx_{\alpha}\rightarrow\cech_{\alpha}$ that connect the two sequences
on a simplicial level. We then show that the induced maps on homology groups all commute
and finally apply Theorem~\ref{thm:c-approx-thm} to show that the constructed module $(1+\eps)$-approximates
the \Cech module.
We remark that this strategy follows the approach by Dey et al.~\cite{dfw-computing} 
who get a similar result for the Rips module,
simplifying the previous work of Sheehy~\cite{sheehy-linear}.

\paragraph{More notations.}
Throughout the section, we assume a finite point set $S\subset\R^d$ and a quadtree $Q$,
and we reuse the notation on quadtrees from the previous section. 
Moreover, we will use assume the existence of an $\frac{\eps}{12}$-WSSD 
defined over cells of $Q$, computed with the algorithm from Section~\ref{sec:wsd}.
We will mostly omit the ``$\frac{\eps}{12}$'' and just talk about the WSSD and its WSTs from now on.
Having a WST $t=(q_0,\ldots,q_k)$, we write $\rad(t)$ for the radius of the 
minimum enclosing ball of $q_0\cup\ldots\cup q_k$.
For a non-empty quadtree cell $q$, we choose a \emph{representative} $\rep(q)$ in $S$ with the property
that if $q$ is internal, its representative is chosen among the representatives of its children.
Moreover, for any quadtree cell $q$ of height $i$ or less, we define $\qcell(q,i)$ 
for its (unique) ancestor at level~$i$.

We fix the following additional parameters:
Set $\theta_\ell:=(1+\frac{\eps}{2})^\ell$ for any integer $\ell$.
Let $\disc_\alpha$ denote the integer such that 
$$\theta_{\disc_\alpha}\leq\alpha < \theta_{\disc_\alpha+1}.$$
Furthermore, we define $\hval_\alpha$ as the integer such that
$$2^{\hval_\alpha}\leq \frac{\eps \theta_{\disc_\alpha}}{3\sqrt{d}} \leq 2^{\hval_\alpha+1}.$$
When there is no ambiguity about $\alpha$, we will skip the suffixes 
and write $\disc:=\disc_\alpha$ and $\hval:=\hval_\alpha$.

To give a rough intuition about the chosen terms, the approximate complex will be only changing at discrete values;
more precisely, all $\alpha\in[\theta_\ell,\theta_{\ell+1})$ will result in the same approximation. This motivates the
definition of $\disc_\alpha$ which determines the range in which $\alpha$ falls in. 
The second parameter $\hval_\alpha$ determines the grid size on which the
approximation is constructed. Note that $\hval_\alpha$ rather depends on $\disc_\alpha$ than on $\alpha$ itself.
Consequently, for any $\alpha\in [\theta_k,\theta_{k+1})$, the same $\hval_\alpha$ is chosen.
Before we formally describe our construction, we prove the following useful lemma:

\begin{lemma}\label{lem:wssd_level_prop}
Let $\alpha>0$, $\disc:=\disc_\alpha$ and $\hval:=\hval_\alpha$ as defined above. If an $\frac{\eps}{12}$-WST $t=(q_0,\ldots,q_\kapa)$ satisfies $\rad(t)\leq\theta_{\disc+1}$,
the height of each $q_i$ is $\hval$ or smaller.
\end{lemma}
\begin{proof}
Since $\mebr(t) \le \theta_{\disc+1}$, Lemma~\ref{lem:heightbnd} implies
that the height $h'$ of each $q_i$ satisfies $2^{h'} \le \frac{\eps \theta_{\disc+1}}{12\sqrt{d}}$.
Note that $\theta_{\disc+1}=(1+\eps/2)\theta_{\disc}\leq 2\theta_{\disc}$, and therefore,
$$2^{h'}\leq \frac{\eps\theta_{\disc}}{6\sqrt{d}}<\frac{2^{\hval+1}}{2}\leq 2^{\hval}.\qedhere$$
\end{proof}

\paragraph{The approximation complex}
Recall that $G_\ell$ denotes the set of all quadtree cell at height $\ell$.
We construct a simplicial complex $\Approx_\alpha$
over the vertex set $G_{\hval}$ (with $\hval:=\hval_\alpha$) in the following way: 
For any WST $t'=(q_0,\ldots,q_k)$ with all $q_i$ at height $\hval$ or less, 
let $t=(\qcell(q_0,\hval),\ldots,\qcell(q_k,\hval))$. If $\rad(t)\leq\theta_\disc$, 
we add the simplex $t$ to $\Approx_\alpha$. 
Note that some of the $\qcell(q_\ell,\hval)$ can be the same, so that 
the resulting simplex might be of dimension less than~$k$.
It is clear by construction and Lemma~\ref{lem:size} 
that $\Approx_\alpha$ consists of at most $n (d/\eps)^{O(d^2)}$ simplices,
but it requires a proof to show that it is well-defined:

\begin{lemma}
$\Approx_\alpha$ is a simplicial complex.
\end{lemma}
\begin{proof}
Let $(q_0,\ldots,q_k)\in\Approx_\alpha$. We need to show that its faces are in $\Approx_\alpha$ as well. 
Wlog consider $(q_0,\ldots,q_\ell)$ with $\ell<k$. Since each $q_i$ is non-empty, we can choose some $v_i\in q_i$ and consider
the simplex $\tau=(v_0,\ldots,v_\ell)$. By the covering property of WSSD, there exists a WST $t'=(q_0',\ldots,q_\ell')$
that covers $\tau$. Note that 
$$\rad(\tau)\leq\rad(q_0,\ldots,q_\ell)\leq\rad(q_0,\ldots,q_k)\leq\theta_{\disc}.$$
Now, because $t'$ is $\frac{\eps}{12}$-well-separated and the meb of $\tau$ intersects all $q_i'$,  
$$\rad(t')\leq (1+\frac{\eps}{12})\rad(\tau)< \theta_{\disc+1}.$$
It follows by Lemma~\ref{lem:wssd_level_prop} that all $q_i'$ are at most on level $\hval$. In particular, for all $i$, 
$\qcell(q_i',\hval)=q_i$ because both cells contain $v_i$, and $q_i$ is on height $h$ by construction. 
Because $\rad(q_0,\ldots,q_\ell)\leq \theta_{\disc}$,
it follows that $(q_0,\ldots,q_\ell)$ belongs to $\Approx_\alpha$ because of the WST $t'$.
\end{proof}

We define maps between the $\Approx_\alpha$ next:
Consider two scales $\alpha_1<\alpha_2$. We set $\hval_1:=\hval_{\alpha_1}$ and define $\hval_2$, $\disc_1$, and $\disc_2$ accordingly.
Since $\hval_1\leq \hval_2$, there is a natural map
$g_{\alpha_1}^{\alpha_2}: G_{\hval_1}\rightarrow G_{\hval_2}$, mapping a quadtree cell at height $\hval_1$ to its ancestor at height $\hval_2$.
This naturally extends to a map $$g_{\alpha_1}^{\alpha_2}:\Approx_{\alpha_1}\rightarrow \Approx_{\alpha_2},$$
by mapping a simplex $\sigma=(v_0,\ldots,v_k)$ to 
$g_{\alpha_1}^{\alpha_2}(\sigma):=(g_{\alpha_1}^{\alpha_2}(v_0),\ldots,g_{\alpha_1}^{\alpha_2}(v_k))$.
It is easy to verify that $g_{\alpha'}^{\alpha''}\circ g_{\alpha}^{\alpha'}=g_{\alpha}^{\alpha''}$ and 
$g_{\alpha}^{\alpha}=\mathrm{id}$.

\begin{lemma}
\label{lem:g_is_simplicial}
$g:=g_{\alpha_1}^{\alpha_2}:\Approx_{\alpha_1}\rightarrow \Approx_{\alpha_2}$ is a simplicial map.
\end{lemma}
\begin{proof}
Let $t=(q_0,\ldots,q_k)$ be a $k$-simplex of $\Approx_{\alpha_1}$. In particular, $\rad(t)\leq\theta_{\disc_1}$ and all cells
are at level $\hval_1$. Let $q_\ell'=g(q_\ell)$ denote the ancestor of $q_\ell$ at level $\hval_2$.
We need to show that $t'=(q_0',\ldots,q_k')\in\Approx_{\alpha_2}$.
For that, it suffices to show that $\rad(t')\leq\theta_{\disc_2}$. Note that $\disc_1=\disc_2$ implies $\hval_1=\hval_2$,
so $t=t'$ and the statement is trivial. So, assume that $\disc_1<\disc_2$.

Consider the minimum enclosing ball of $t$. 
Note that this ball contains $q_\ell$, and therefore also at least one point of each $q_\ell'$, for $0\leq\ell\leq k$.
We increase the radius by (at least) the diameter of a quadtree cell on level $\hval_2$.
The enlarged ball then contains $q_\ell'$ completely (compare Observation~\ref{obs:offset}). 
The diameter of the cells at level $\hval_2$, however, is at most
$$\frac{\eps\theta_{\disc_2}}{3\sqrt{d}}\sqrt{d}\leq \frac{\eps}{3}\theta_{\disc_2}.$$
Moreover, because $\disc_1$ is strictly smaller than $\disc_2$, 
$\theta_{\disc_1}\leq\frac{\theta_{\disc_2}}{1+\frac{\eps}{2}}$.
It follows that
$$\rad(t')\leq \rad(t)+\frac{\eps}{3}\theta_{\disc_2}\leq \frac{1+\frac{\eps}{3}+\frac{\eps^2}{6}}{1+\frac{\eps}{2}}\theta_{\disc_2}\leq\theta_{{\disc_2}}$$
for all $\eps\leq 1$. Therefore, $t'\in\Approx_{\alpha_2}$.
\end{proof}

\paragraph{Cross maps}
Next, we investigate the \emph{cross-map} $\crossmapdown:\cech_{\frac{\alpha}{1+\eps}}\rightarrow\Approx_{\alpha}$. To define it
for a vertex $v\in\cech_{\frac{\alpha}{1+\eps}}$ (which is a point of $S$), 
set $\crossmapdown(v)=q$, where $q$ is the quadtree cell at level $\hval$ that contains $v$.
For a simplex $(v_0,\ldots,v_k)$, define $\crossmapdown(v_0,\ldots,v_k)=(\crossmapdown(v_0),\ldots,\crossmapdown(v_k))$.

\begin{lemma}
$\crossmapdown$ is a simplicial map.
\end{lemma}
\begin{proof}
Fix a simplex $\sigma=(v_0,\ldots,v_k)\in\cech_{\frac{\alpha}{1+\eps}}$. 
Take a WST $t=(q_0,\ldots,q_k)$ that covers $\sigma$. By the properties of the $\frac{\eps}{12}$-WSSD, it follows that
$$\rad(t)\leq (1+\frac{\eps}{12})\rad(\sigma)\leq \frac{(1+\frac{\eps}{12})}{1+\eps}\alpha.$$
Now, since $\frac{1+\frac{\eps}{12}}{1+\eps}\alpha<\alpha\leq\theta_{\disc+1}$, we can apply Lemma~\ref{lem:wssd_level_prop}
which guarantees that all $q_\ell$ are at level at most $\hval$. Let $t'=(q_0',\ldots,q_k')$ with $q_\ell'=\qcell(q_\ell,\hval)$.
Note that $q_\ell'=\crossmapdown(v_\ell)$, so all we need to show is that $t'\in \Approx_{\alpha}$.
As explained in the proof of Lemma~\ref{lem:g_is_simplicial}, the diameter of a cell at level $\hval$
is at most $\frac{\eps}{3}\theta_\disc$. It follows
that the minimum enclosing ball of $t$ enlarged by $\frac{\eps}{3}\theta_\disc$ covers $t'$. Therefore,
$$\rad(t')\leq \rad(t)+\frac{\eps}{3}\theta_\disc\leq \frac{(1+\frac{\eps}{12})}{1+\eps}\alpha+\frac{\eps}{3}\theta_\disc$$
Since $\alpha\leq (1+\frac{\eps}{2})\theta_\disc$, this implies
$$\rad(t')\leq\frac{1+\frac{5}{12}\eps+\frac{3}{8}\eps^2}{1+\eps}\theta_\disc \leq \theta_\disc$$
for $\eps\leq\frac{14}{9}$. It follows that $t'\in\Approx_{\alpha}$.
\end{proof}

In the other direction, we have a map $\crossmapup:\Approx_{\alpha}\rightarrow\cech_{\alpha}$
defined by mapping a quadtree cell $q$ at level $\hval$ to its representative $\rep(q)$.
It is easy to see that this map is simplicial: For $t=(q_0,\ldots,q_k)$ in $\Approx_{\alpha}$, we have that
$\rad(t)\leq\theta_\disc\leq\alpha$. Setting $\sigma:=(\rep(q_0),\ldots,\rep(q_k))$, it is clear that 
$\rad(\sigma)\leq\rad(t)\leq\alpha$, so $\sigma\in\cech_\alpha$.

\paragraph{Interleaving  sequences}
We fix some integer $p\geq 0$ and consider the persistence modules
$$(\hom{\cech}_\alpha)_{\alpha\geq 0}:=(H_p(\cech_\alpha))_{\alpha\geq 0},\quad (\hom{\Approx}_\alpha)_{\alpha\geq 0}:=(H_p(\Approx_\alpha))_{\alpha\geq 0},$$
where $H_p(\cdot)$ is the $p$-th homology group over an arbitrary base field, with the induced homomorphisms
 $\hom{f}_{\alpha_1}^{\alpha_2}$ (induced by inclusion) and $\hom{g}_{\alpha_1}^{\alpha_2}$, respectively.
Moreover, since the cross-maps are simplicial, the induced homomorphisms $\hom{\crossmapdown}:\hom{\cech}_{\frac{\alpha}{1+\eps}}\rightarrow\hom{\Approx}_{\alpha}$
and $\hom{\crossmapup}:\hom{\Approx}_\alpha\rightarrow \hom{\cech}_{\alpha}$
connect the two modules.
We show that the cross-maps $\hom{\crossmapdown}$, $\hom{\crossmapup}$ commute 
with the module maps $\hom{f}$, $\hom{g}$ in the next three lemmas.

\begin{lemma}\label{lem:com_1}
The diagram
$$
\xymatrix{
\hom{\cech}_{\frac{\alpha}{1+\eps}} \ar[rd]^{\hom{\crossmapdown}} & \\
\hom{\Approx}_{\frac{\alpha}{1+\eps}} \ar[u]^{\hom{\crossmapup}} \ar[r]^{\hom{g}} & \hom{\Approx}_{\alpha} \\
}
$$
commutes, that means, $\hom{\crossmapdown}\circ\hom{\crossmapup}=\hom{g}$.
\end{lemma}
\begin{proof}
The maps commute already on the simplicial level, that is, $\crossmapdown\circ\crossmapup=g$, as one can easily verify from the definition of the maps.
\end{proof}

For the next two lemmas, we need the following definition: Two simplicial maps $h_1,h_2:K\rightarrow L$ 
are \emph{contiguous} if for any simplex $(v_0,\ldots,v_k)\in K$, the points $(h_1(v_0),\ldots,h_1(v_k),h_2(v_0),\ldots,h_2(v_k))$
form a simplex in $L$. In this case, the induced homomorphisms $\hom{h_1}, \hom{h_2}$ are equal~\cite[p.67]{munkres}.

\begin{lemma}\label{lem:com_2}
The diagram
$$
\xymatrix{
\hom{\cech}_{\frac{\alpha}{1+\eps}} \ar[rd]^{\hom{\crossmapdown}} \ar[r]^{f} & \hom{\cech}_{\alpha}\\
& \hom{\Approx}_{\alpha} \ar[u]^{\hom{\crossmapup}} \\
}
$$
commutes, that means, $\hom{\crossmapup}\circ\hom{\crossmapdown}=\hom{f}$.
\end{lemma}
\begin{proof}
Note the simplicial maps do not commute here; we will show instead that they are contiguous.
So, fix a simplex $\sigma=(v_0,\ldots,v_k)$ in $\cech_{\frac{\alpha}{1+\eps}}$. 
Consider its image $(q_0,\ldots,q_k)$ under $\crossmapdown$. All $q_\ell$ are on level $\hval$,
$v_\ell \in q_\ell$,
and $\rad(q_0,\ldots,q_k)\leq \theta_\disc\leq\alpha$. Let $(w_0,\ldots,w_k)$ be the image
of $(u_0,\ldots,u_k)$
under $\crossmapup$, that is, $w_\ell$ is the representative of $q_\ell$.
In particular, we have that $w_\ell\in q_\ell$. It follows that the set $\{v_0,\ldots,v_k,w_0,\ldots,w_k\}$
is contained in the union $q_0\cup\ldots\cup q_k$ and therefore, 
$\rad(v_0,\ldots,v_k,w_0,\ldots,w_k)\leq\alpha$. It follows that the simplex $(v_0,\ldots,v_k,w_0,\ldots,w_k)$
is in $\cech_\alpha$. Hence, $\crossmapup\circ\crossmapdown$ and $f$ are contiguous.
\end{proof}

\begin{lemma}\label{lem:com_3}
For $\alpha_1\leq\alpha_2$, the diagram
$$
\xymatrix{
\hom{\cech}_{\alpha_1} \ar[r]^{\hom{f}} & \hom{\cech}_{\alpha_2}\\
\hom{\Approx}_{\alpha_1} \ar[r]^{\hom{g}} \ar[u]^{\hom{\crossmapup}}& \hom{\Approx}_{\alpha_2} \ar[u]^{\hom{\crossmapup}} \\
}
$$
commutes, that means, $\hom{\crossmapup}\circ \hom{g}=\hom{f}\circ\hom{\crossmapup}$.
\end{lemma}
\begin{proof}
Again, the corresponding simplicial maps do not commute in general (they do only if $\hval_{\alpha_1}=\hval_{\alpha_2}$).
We will show that the simplicial maps are contiguous. Fix some $t=(q_0,\ldots,q_k)\in\Approx_{\alpha_1}$
and let $v_\ell$ be the representative of $q_\ell$; in particular $f\circ\crossmapup(q_\ell)=v_\ell$.
Now, set $q_\ell':=g(q_\ell)$. It is clear that $q_\ell\subseteq q_\ell'$. Moreover, by definition
of $\Approx_{\alpha_2}$, we have that $\rad(q_0',\ldots,q_k')\leq \theta_{\disc_2}\leq\alpha_2$,
where $\disc_2:=\disc_{\alpha_2}$. Set $w_\ell:=\crossmapup(g(q_\ell))=\crossmapup(q_\ell')$ be the representative of $q_\ell'$.
By construction, $v_0,\ldots,v_k,w_0,\ldots,w_k$ are all contained in the union $q_0'\cup\ldots\cup q_k'$ and therefore,
$\rad(v_0,\ldots,v_k,w_0,\ldots,w_k)\leq\alpha_2$. This implies that the two maps are contiguous.
\end{proof}

\begin{theorem}
The persistence module $\hom\Approx_\ast$ is a $(1+\eps)$-approximation of the persistence module $\hom{\cech}_\ast$.
\end{theorem}
\begin{proof}
Using Lemmas~\ref{lem:com_1}-\ref{lem:com_3}, one can show that all diagrams in~\eqref{eqn:commuting} commute by splitting them into subdiagrams.
The result follows from Theorem~\ref{thm:c-approx-thm}.
\end{proof}

\section{Coresets for minimal enclosing ball radii}
\label{sec:coreset}
Recall that for a point set $P=\{p_1,\ldots,p_n\}\subset\R^d$, we denote by
$\meb(P)$ the \emph{minimum enclosing ball of $P$}. 
Let $\mebc(P)\in\R^d$ denote the center and $\mebr(P)\geq 0$ the radius
of $\meb(P)$. Fixing $\eps>0$, 
we call a subset $C\subseteq P$ a \emph{meb-coreset for $P$}
if the ball centered at $\mebc(C)$ and with radius $(1+\eps)\mebr(C)$
contains $P$. We call $C\subseteq P$ a \emph{radius-coreset for $P$}
if $\mebr(P)\leq (1+\eps)\mebr(C)$. Informally, a radius-coreset
approximates only the radius of the minimum enclosing ball, whereas
the meb-coreset approximates the ball itself. 
A meb-coreset is also a radius-coreset
by definition, but the opposite is not always the case;
see Figure~\ref{fig:coreset_illu} for an example.

\begin{figure}[htb]
\centering
\includegraphics[width=4cm]{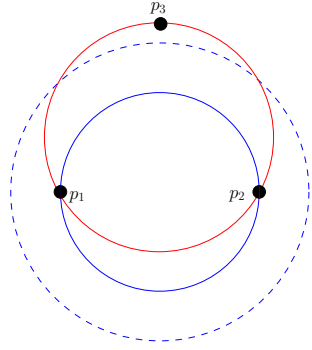}
\caption{Consider the equilateral triangle with points $p_1=(-1,0)$, $p_2=(1,0)$
and $p_3=(0,\sqrt{3})$ in the plane; 
let $P=\{p_1,p_2,p_3\}$ and $C=\{p_1,p_2\}$.
Then, $\mebc(P)=(0,\sqrt{1/3})$, $\mebr(P)=\sqrt{\frac{4}{3}}$, $\mebc(C)=(0,0)$,
and $\mebr(C)=1$. For $\eps=0.5$, it is thus clear that $C$ is a radius-coreset
of $P$. However, $C$ is not a meb-coreset because the ball with radius $1.5$
around the origin does not contain $p_3$.}
\label{fig:coreset_illu}
\end{figure}

Obviously, a point set is a coreset of itself, so coresets
exist for any point set. We are interested in the coresets of small sizes.
For the meb-coreset, this question is answered by B\u{a}doiu and Clarkson~\cite{bc-optimal}. We summarize their result in the following statement:
\begin{theorem}
For $\eps>0$, and any (finite) point set, there exists
a meb-coreset of size $\lceil \frac{1}{\eps}\rceil$, and there exist
point sets where any meb-coreset has size at least 
$\lceil \frac{1}{\eps}\rceil$.
\end{theorem}
Note that the size of the coreset is independent of both the number of points
in $P$ and the ambient dimension. However, since radius-coresets are a relaxed
version of meb-coresets, we can hope for even smaller coresets.
We start by showing a lower bound:

\begin{lemma}
There is a point set such that any radius-coreset 
has size at least
$$\delta:=\lceil\frac{1}{2\eps+\eps^2}+1\rceil.$$
\end{lemma}
\begin{proof}
Consider the standard $(d-1)$-simplex in $d$ dimensions,
that is, $P$ is the point set given by the $d$ unit vectors in $\R^d$.
By elementary calculations, it can be verified that 
$\mebc(P)=(\frac{1}{d},\ldots,\frac{1}{d})$
and $\mebr(P)=\sqrt{\frac{d-1}{d}}$. 
Fixing a subset $C\subseteq P$ of size $k$, 
its points span a standard simplex in $\R^k$
and therefore, $\mebr(C)=\sqrt{\frac{k-1}{k}}$ by the same argument.
Hence, $C$ is a radius-coreset of $P$ if and only if
$$\sqrt{\frac{d-1}{d}}\leq (1+\eps)\sqrt{\frac{k-1}{k}}.$$
Isolating $k$ yields the equivalent condition that
$$k\geq \lceil\frac{(1+\eps)^2}{(1+\eps)^2-\frac{d-1}{d}}\rceil=\lceil1+\frac{1}{\frac{d}{d-1}(2\eps+\eps^2+\frac{1}{d})}\rceil.$$
The last expression is monotonously increasing in $d$,
and converges to $\delta$ for $d \rightarrow\infty$.
It follows that, for $d$ large enough, any radius-coreset of a standard $(d-1)$-simplex
has size at least~$\delta$.
\end{proof}

We will show next that any point set has a radius coreset of size $\delta$.
For a point set $P$ in $\R^d$ and $1\leq k\leq d$, let $r_k(P)$ denote
the maximal radius of a meb among all subsets of $P$ of cardinality $k$.
We can assume that $P$ contains at least $d+1$ points; otherwise it is
contained in a lower-dimensional Euclidean space. On the other hand,
if $P$ contains at least $d+1$ points,
there exists a subset $P'$ of $P$ containing exactly $d+1$ points
such that the meb of $P'$ equals the meb of $P$, which implies that
$r_{d+1}(P)=\mebr(P)$. Moreover, $r_2(P)=\diam(P)$ is the diameter of $P$.
We use a result by Henk~\cite[Thm.1]{henk-generalization} (we adapt
his notation to our context):
\begin{theorem}[Generalized Jung's Theorem]\label{thm:henk}
Let $P\subset\R^d$ be a point set, and let $i$, $j$ two integers with 
$2\leq j\leq i\leq d+1$. Then
$$r_i(P)\leq \sqrt{\frac{j(i-1)}{i(j-1)}}r_j(P)$$
\end{theorem}

The theorem generalizes an older result by Jung~\cite{jung-ueber} which states
the following relation between the circumradius and the diameter of $P$:
\begin{eqnarray}
\mebr(P)=r_{d+1}(P)\leq\sqrt{\frac{2d}{d+1}}r_2(P)=\sqrt{\frac{2d}{d+1}}\diam(P).\label{eqn:jung}
\end{eqnarray}
We sketch the proof of Theorem~\ref{thm:henk} for completeness.
It relies on the following property: Given a 
point set $Q$ of $k+1$ linearly independent points in $\R^k$. Then,
\begin{eqnarray}
\mebr(Q)\leq \frac{k}{\sqrt{k^2-1}}r_k(Q),\label{eqn:face_lemma}
\end{eqnarray}
in other words, there is a subset of $k$ points whose circumradius
is large in some sense; see also~\cite[Lemma 3.3]{bc-optimal}.
We assume for simplicity that the $i$-subset of points of $P$
that realizes $r_i(P)$ is linearly independent; otherwise, we can switch
to an independent subset and a similar argument applies.
Iteratively applying \eqref{eqn:face_lemma} yields that
$$r_i(P)\leq \prod_{t=j}^{i-1}\frac{t}{\sqrt{t^2-1}} r_j(P).$$
However, it is a straight-forward to prove by induction that
$$\prod_{t=j}^{i-1}\frac{t}{\sqrt{t^2-1}} = \sqrt{\frac{j(i-1)}{i(j-1)}}.$$

\begin{theorem}\label{thm:coreset}
For $\eps>0$, any finite point set $P$ has a radius-coreset of size $\delta$.
\end{theorem}
\begin{proof}
Applying Theorem~\ref{thm:henk} to the case that $i=d+1$ and $j=\delta$ yields
$$\mebr(P)=r_{d+1}(P)\leq \sqrt{\frac{\delta\cdot d}{(d+1)(\delta-1)}}r_{\delta}(P)=\underbrace{\sqrt{\frac{d}{d+1}}}_{\leq 1}\sqrt{\frac{\delta}{\delta-1}}r_{\delta}(P).$$
Furthermore, since $\delta\geq \frac{1}{2\eps+\eps^2}+1$,
it follows that
$$\frac{\delta}{\delta-1}=1+\frac{1}{\delta-1}\leq (1+\eps)^2.$$
So, letting $C$ be a subset of cardinality $\delta$ with radius $r_\delta(P)$,
we obtain that $\mebr(P)\leq (1+\eps)\mebr(C)$, 
which means that $C$ is a radius-coreset.
\end{proof}
We remark that our results immediately imply an algorithm for computing
a radius-coreset of size $\delta$: 
starting with the whole point set, iteratively remove points such that
the remaining subset has the largest possible radius among all choices
of removed points. When this process is stopped for a subset of size $\delta$,
the resulting subset is a radius-coreset. However, this algorithm
is rather inefficient, because it is quadratic in $n$, and a natural
question is how to compute radius coresets more efficiently.
For meb-coresets of size $\lceil\frac{1}{\eps}\rceil$, 
B\u{a}doiu and Clarkson~\cite{bc-optimal} prove existence
algorithmically by defining an algorithm which starts
with an arbitrary set of size $\lceil\frac{1}{\eps}\rceil$ and 
alternatingly adds and removes points from the set until the set 
remains unchanged, and they prove that the resulting set is a meb-coreset.
Their algorithm is an instance of a more general class of optimization
problems as described in~\cite{clarkson-coresets}; we were not able to
find a reformulation of the radius-coreset problem in terms
of this algorithmic framework.

\section{A generalized Rips-Lemma}
\label{sec:clique}

We define the following generalization of a flag-complex:
\begin{definition}[$i$-completion]
Let $K$ denote a simplicial complex. 
The \emph{$i$-completion} of $K$, $\compl_i(K)$, 
is maximal complex whose $i$-skeleton
equals the $i$-skeleton of~$K$.
\end{definition}

With that notation, we have that
$\rips_\alpha = \compl_1(\cech_\alpha)$.
Moreover, we have that
$\cech_\alpha=\compl_d(\cech_\alpha)$ as a consequence of
Helly's Theorem~\cite[p.57]{eh-computational}.

\makespace

We can show the following result as an application of 
Theorem~\ref{thm:coreset}.

\begin{theorem}\label{thm:interleaving}
For $\delta=\lceil 1/(2\eps+\eps^2)+1\rceil$,
$$\cech_\alpha\subseteq \compl_{\delta-1}(\cech_\alpha) \subseteq \cech_{(1+\eps)\alpha}$$
\end{theorem}
\begin{proof}
The first inclusion is clear. Now, consider a simplex $\sigma$ in
$\compl_{\delta-1}(\cech_\alpha)$. The second inclusion is trivial if
$\dim \sigma\leq\delta-1$, so let its dimension be at least $\delta$.
By Theorem~\ref{thm:coreset}, 
the boundary vertices of $\sigma$ have a coreset of
size at most $\delta$. Let $\tau$ denote the simplex spanned by such a 
coreset. As $\tau$ is a face of $\sigma$, it
is contained in $\compl_{\delta-1}(\cech_\alpha)$,
and because it is of dimension at most $\delta-1$, it is in particular contained
in $C(\alpha)$. By the property of coresets, the minimal enclosing ball
of $\sigma$ has radius at most $(1+\eps)\alpha$ which implies
that $\sigma\in\cech_{(1+\eps)\alpha}$.
\end{proof}

As a special case, consider the choice $\eps=\sqrt{2}-1$, so that $\delta=2$. 
The above result yields that
$$\rips_\alpha=\compl_{1}(\cech_\alpha)\subseteq \cech_{\sqrt{2}\alpha},$$
which is exactly the statement of the Vietoris-Rips Lemma
as stated in \cite[p.62]{eh-computational}.


Theorem~\ref{thm:interleaving} and Lemma~\ref{lem:pat} prove the
closeness of the persistence diagrams of the \Cech filtration
and the completion complex:
\begin{theorem}\label{thm:completion_theorem}
The persistence diagram of $\compl_{\delta-1}(\cech_\ast)$ with
$\delta:=\lceil1/(2\eps+\eps^2)+1\rceil$
is a $(1+\eps)$-approximation of the persistence diagram
of $\cech_{\ast}$.
\end{theorem}

Note that $\compl_{k}(\cech_\alpha)$ is determined by the $k$-skeleton of the \Cech complex,
which of size $O(n^{k+1})$. In this respect, the completion complex constitutes a trade-off
between simplicity (i.e., its representation size) and approximation quality of the \Cech complex.
We emphasize that the approximation is solely determined by $k$ and does not depend on the ambient
dimension of the point set.

\section{Conclusion and Outlook}
\label{sec:conclusion}
We have presented two distinct ways to approximate \Cech complexes;
the fixed-dimensional result on approximating the \Cech filtration
to linear size is a technically challenging, but
conceptually straight-forward extension of recent work on the Rips filtration;
however, we believe that the concept of WSSDs to be interesting
and hopefully applicable in different contexts, and we plan to identify
application scenarios in the future.
Our high-dimensional
results are a first attempt to link the areas of computational topology,
where data is often high-dimensional, and geometric approximation algorithms
that try to overcome the curse of dimensionality.
We want to achieve algorithmic results in that context in the future;
one question is whether an optimal-size radius coreset can be computed efficiently.
Moreover, the introduced concept of completions is not tied to start
completing simplices at a fixed dimension;
in fact, one can start
with any complex $C$ (not necessarily a skeleton) 
and define the completion as the largest complex containing $C$.
With such \emph{adaptive completions}, $\eps$-close approximations of the \Cech filtration
might be possible with just a slightly larger representation size than the
Rips filtration. The open question is, however, whether such a representation
can be computed efficiently.
Finally, we pose the question whether there are other applications,
besides approximating \Cech complexes, where the smaller size of radius-coresets
in comparison to meb-coresets could be useful.

\end{document}